\documentclass[conference, 9pt]{IEEEtran}
\usepackage[utf8]{inputenc}

\usepackage{amsmath,amssymb, amsthm, mathtools, braket}
\usepackage{braket}
\usepackage[pdftex]{graphicx}
\usepackage{caption,subcaption} % depreciated, but still used in  IEEEtran
\usepackage{enumitem}
\usepackage{algorithm, algorithmic}

\usepackage{geometry}
\usepackage{bm}

\long\def\comment#1{}

\newfont{\bbb}{msbm10 scaled 700}

\newcommand{\bv}{{\bf b}}

\newcommand{\ev}{{\bf e}}
\newcommand{\fv}{{\bf f}}

\newcommand{\tv}{{\bf t}}

\newcommand{\xv}{{\bf x}}
\newcommand{\yv}{{\bf y}}

\newcommand{\zerov}{{\bf 0}}

\newcommand{\Am}{{\bf A}}
\newcommand{\Bm}{{\bf B}}

\newcommand{\Hm}{{\bf H}}

\newcommand{\Pm}{{\bf P}}

\newcommand{\Sm}{{\bf S}}
\newcommand{\Tm}{{\bf T}}

\newcommand{\Hc}{{\cal H}}

\newcommand{\Sc}{{\cal S}}
\newcommand{\Tc}{{\cal T}}

\newtheorem{theorem}{Theorem}

\begin{document}

\title{Guided Signal Reconstruction with Application to \\
Image Magnification %GlobalSIP15
\vspace{-0.25\baselineskip}}
\author{\IEEEauthorblockN{Akshay Gadde}
\IEEEauthorblockA{University of Southern California\\
%line 2: name of organization, acronyms acceptable\\
%line 3: City, State/Province, Country\\
%line 4: e-mail address if desired\\
agadde@usc.edu
\vspace{-2\baselineskip}
}
\and
\IEEEauthorblockN{Andrew Knyazev, Dong Tian, Hassan Mansour}
\IEEEauthorblockA{Mitsubishi Electric Research Laboratories\\
%line 2: name of organization, acronyms acceptable\\
%line 3: City, State/Province, Country\\
%line 4: e-mail address if desired
\{knyazev, tian, mansour\}@merl.com
\vspace{-2\baselineskip}
}}

%\author{%Anonymous
%Akshay Gadde,  Andrew Knyazev, Dong Tian, Hassan Mansour
%}

\maketitle

\begin{abstract}
We study the problem of reconstructing a signal from its projection on a subspace.
The proposed signal reconstruction algorithms utilize a guiding subspace that represents desired properties of reconstructed signals.
We show that optimal reconstructed signals belong to a convex bounded set, called the ``reconstruction'' set. 
We also develop iterative algorithms, based on conjugate gradient methods, 
to approximate optimal reconstructions with low memory and computational costs. 
The effectiveness of the proposed approach is demonstrated for image magnification, 
where the reconstructed image quality is shown to exceed that of consistent and generalized reconstruction schemes.
\end{abstract}

%\vspace{-0.2\baselineskip}
\section{Introduction}
%\vspace{-0.2\baselineskip}
The problem of reconstructing a signal from its partial observations is of fundamental importance in signal processing. 
A classical example is reconstruction of continuous bandlimited signals from their discrete time samples. 
There are numerous applications of signal reconstruction, e.g.,\ 
%in audio, image and video processing, and machine learning. Examples include 
image super-resolution~\cite{park2003super}, increasing audio frequency range~\cite{bansal2005bandwidth}, and semi-supervised learning~\cite{Gadde-KDD-14}.

We consider the problem of determining a reconstruction $\hat{\fv}$ of an original signal $\fv$ from its measurement obtained by taking an orthogonal projection $\Sm\fv$ onto a subspace $\Sc$ called the \emph{sampling subspace}. 
Since sampling involves loss of information, we need some \emph{a priori} assumptions on the original signal $\fv$ to be recovered. One such assumption may be that the signal $\fv$ belongs to a subspace $\Tc$ (of a vector space $\Hc$) that can be thought of as a \emph{target reconstruction subspace}.
We prefer to call $\Tc$ a \emph{guiding reconstruction subspace}, because, even though we expect $\fv$ to have most of its energy contained in $\Tc$, in our technique the reconstructed signal $\hat{\fv}$ is not necessarily restricted to $\Tc$. 
The guiding subspace of the signal can be determined using some model of desirable reconstructed signal behavior.
For example, it can be learned from a training dataset~\cite{bansal2005bandwidth}.
For signals in euclidean spaces with natural spectral properties, a space of signals bandlimited in the transform domain (such as Fourier, cosine and wavelet) can be chosen as the guiding subspace. This idea can be extended to signals defined on manifolds or graphs, where $\Tc$ can be chosen as the space of smooth signals given by the linear combinations of first few eigenvectors of the Laplacian operator associated with the manifold or the graph~\cite{Anis-ICASSP-14}.

A reconstruction $\hat{\fv}$ is said to be sample consistent if $\Sm\hat{\fv} = \Sm\fv$, i.e., the measurements remain unchanged. 
A set of all signals, having the same samples $\Sm\fv$ is a plane $\Sm\fv+\Sc^\perp$
that we call a \emph{sample consistent plane},
where $\Sc^\perp$ is the orthogonal complement to the sampling subspace $\Sc$.
The sets $\Sm\fv+\Sc^\perp$ and $\Tc$ in general may not intersect. In such a case, there is no sample consistent reconstruction which is also in $\Tc$. For a solution, which is in both $\Tc$ and $\Sm\fv + \Sc^\perp$ to exist for any $\fv$, we need $\Tc + \Sc^\perp = \Hc$. 
Additionally, for such a solution to be unique we need $\Tc \cap \Sc^\perp = \{\zerov\}$. 
If both the existence and uniqueness conditions are satisfied, then a unique sample consistent solution in $\Tc$ is given by $\Pm_{\Tc\perp\Sc}\fv$, where $\Pm_{\Tc\perp\Sc}$ is an oblique projector on $\Tc$ along $\Sc^\perp$~\cite{Unser-TSP-94, Eldar-JFA-03}.
Non-uniqueness caused by $\Tc \cap \Sc^\perp \neq \{\zerov\}$ can be mathematically resolved by replacing $\Hc$ with a quotient space $\Hc/\{\Tc\cap \Sc^\perp\}$. 
In practice, one can choose a unique solution by imposing additional constraints; see, e.g.~\cite{Hirabayashi-TSP-07}. 

The assumption $\Tc + \Sc^\perp = \Hc$ can be disadvantageous and very restrictive in applications. Even though it guarantees the existence of the intersection of $\Sm\fv+\Sc^\perp$ and $\Tc$, finding this intersection numerically may be difficult as it is very sensitive to their mutual position (especially in high dimensions). It can lead to oblique projectors with large norms and make the reconstruction unstable.

To counter this, oversampling is advocated \cite{adcock2012generalized}, leading to a smaller consistent plane $\Sm\fv+\Sc^\perp$ that may no longer intersect with $\Tc$.
In \cite{Berger-2013}, the reconstructed signal is defined as a point in $\Tc$ having the smallest distance to $\Sm\fv+\Sc^\perp$ by enforcing the constraint $\hat{\fv}\in\Tc$ and allowing $\hat{\fv}$ to be sample inconsistent. This is known as ``generalized reconstruction''. Generalized reconstruction is more stable. However, the stability comes at the cost of potentially producing a sample inconsistent signal.
In contrast, \cite{bansal2005bandwidth} describes a reconstruction method, which places the reconstructed signal in the sample consistent plane, relaxing the constraint that $\hat{\fv}\in\Tc$ by minimizing the energy of the reconstruction in $\Tc^\perp$ (the orthogonal complement of $\Tc$). Thus, the reconstructed signal is a point in $\Sm\fv+\Sc^\perp$ having the smallest distance to $\Tc$.
This approach is motivated by a realization that in practice it is hard to find a subspace $\Tc$ such that the signals of interest are completely contained in it. Thus, the subspace $\Tc$ is used as a guide, not as a true target, placing the trust on sampling. However, this may not be desirable, if the samples are noisy.

In this paper, we provide a unified view of signal reconstruction under the oversampling scenario. In Section~\ref{s:rs}, we define a set of reconstructions given by the convex combinations of generalized reconstruction and consistent reconstruction and show that it lies on the shortest pathway between the consistent hyperplane and the guiding subspace.
A novel formulation of the sample consistent reconstruction in case of oversampling is provided in Section~\ref{s:pra}.
We present a conjugate gradient (CG) based iterative method to find the sample consistent reconstruction. Our formulation allows an implicit frame-less description of the guiding subspace via an action of the corresponding orthogonal projector $\Tm$ (that can be approximate) and also causes CG iterations to converge faster. This solution can then be projected on $\Tc$ to obtain generalized reconstruction.
Section~\ref{sec:regularization} connects our reconstruction set to a solution of a regularization-based reconstruction problem, cf.~\cite{Narang-GlobalSIP-13}, 
that minimizes the weighted sum of a sample consistency and a smoothness term. %and, thus, provides an efficient way to deal with noisy samples.
We apply our approach to image magnification in Section \ref{s:s}, 
where the reconstructed image quality is shown to exceed that of consistent and generalized reconstructions.%\vspace{-1\baselineskip}

%\newpage

\section{Reconstruction Set}\label{s:rs}
%\vspace{-0.2\baselineskip}
We consider the problem of reconstruction in the oversampling scenario. As stated before, the guiding set (or subspace) may not contain any sample consistent solutions. When the samples are noisy, the original signal does not lie in the sample consistent plane. On the other hand, the original signal may not be entirely contained in the guiding subspace either. In such a case, it is not clear which reconstruction, consistent~\cite{bansal2005bandwidth} or generalized~\cite{adcock2012generalized, Berger-2013}, is better. This situation is illustrated in Fig.~\ref{fig:1} by a simple geometric example, where $\dim\Hc=3$, $\dim\Sc=2$ and $\dim\Tc=1.$ 
Here, the set of all signals having the same 
sample $\Sm\fv$ is a line $\Sm\fv+\Sc^\perp$.
The lines $\Sm\fv+\Sc^\perp$ and $\Tc$ in general do not intersect and no reconstruction $\hat{\fv}$ can be constrained to 
both lines as required in~\cite{Unser-TSP-94, Eldar-JFA-03}. 
        \begin{figure}
        \centering
                \includegraphics[width=0.6\linewidth]{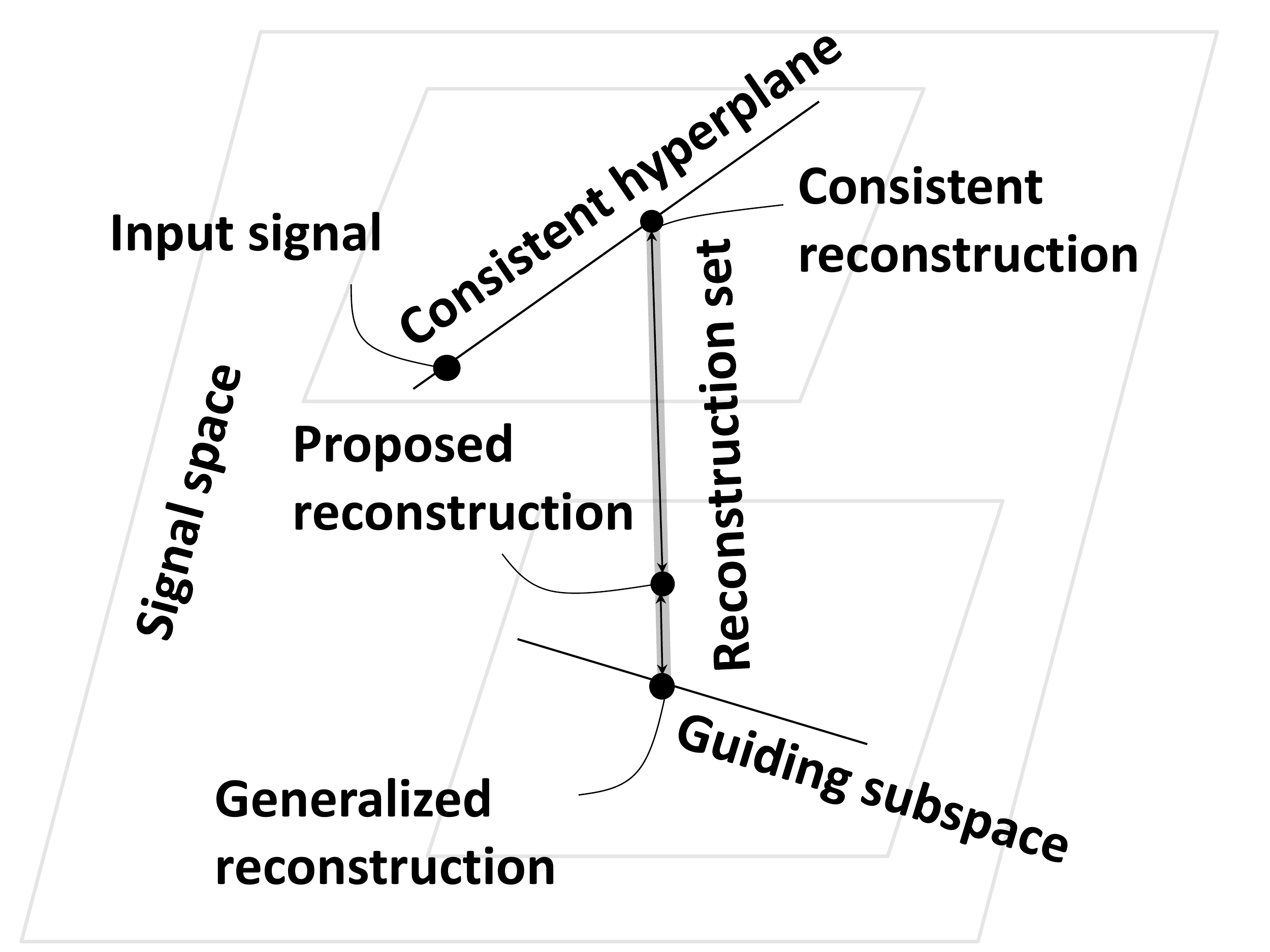}
                \caption{An example with $\dim  \Hc = 3, \dim \Sc = 2,  \dim \Tc = 1$}\label{fig:1}
                \vspace{-1\baselineskip}
        \end{figure}

We see in Fig.~\ref{fig:1} that the consistent reconstruction of \cite{bansal2005bandwidth} can be viewed as an element \emph{from} the consistent plane $\Sm\fv+\Sc^\perp$ which minimizes the distance \emph{to} the guiding subspace $\Tc$. On the other hand, the generalized reconstruction of \cite{Berger-2013} is an element \emph{from} the guiding subspace $\Tc$, minimizing the distance \emph{to} the consistent plane $\Sm\fv+\Sc^\perp$. Clearly, the following equalities hold:
%\small 
\begin{equation}
\min_{\hat\fv\in\Sm\fv+\Sc^\perp}\min_{\tv\in\Tc} \|\hat{\fv}-\tv\|=
\min_{\substack{\hat{\fv}\in\Sm\fv+\Sc^\perp\\ \tv\in\Tc}} \|\hat{\fv}-\tv\|
=\min_{\tv\in\Tc}\min_{\hat{\fv}\in\Sm\fv+\Sc^\perp} \|\hat{\fv}-\tv\|.
\label{eq:recon_problem_geom}
\end{equation}
The above equations suggest defining  a \emph{reconstruction set} as a shortest pathway set 
\emph{between} the consistent plane $\Sm\fv+\Sc^\perp$ and the guiding subspace~$\Tc$.
In Fig.~\ref{fig:1}, the reconstruction set is a line segment with the end points given by the consistent reconstruction
and the generalized reconstruction. Any element of the reconstruction set is a valid candidate for reconstruction when the sampling or guiding procedures are not known to be reliable. Section~\ref{sec:regularization} shows an example of selecting an optimal solution from the reconstruction set when sampling is noisy and the amount of noise is known. 

\section{Algorithm for finding the reconstruction set}\label{s:pra}
We now propose a novel algorithm for finding the sample consistent reconstruction~\cite{bansal2005bandwidth}, which relaxes the constraint that $\hat{\fv} \in \Tc$ and instead minimizes the energy in $\Tc^\perp$ while maintaining sample consistency, as follows, \begin{equation}
\inf_{\hat{\fv}} \|\hat{\fv}-\Tm\hat{\fv}\| \text{ subject to } \Sm\hat{\fv} = \Sm\fv.
\label{eq:recon_problem}
\end{equation}
The above is equivalent to the problem  
\begin{equation}
\inf_{\hat\xv \in \Sc^\perp} \Braket{\left(\hat\xv+\Sm\fv\right), \Tm^\perp \left(\hat\xv+\Sm\fv\right)},
\label{eq:opt_energy}
\end{equation}
where $\hat\xv=\hat\fv-\Sm\fv$.
If the solutions $\hat{\fv}$ and $\hat{\xv}$ to problems  \eqref{eq:recon_problem} and  \eqref{eq:opt_energy}, respectively, are not unique, we choose solutions in the corresponding factor spaces, e.g., the normal solution (i.e., with the smallest norm), to guarantee uniqueness.
%Least squares minimization formulations \eqref{eq:recon_problem} and  \eqref{eq:opt_energy} have an elegant geometric interpretation (as illustrated before in Fig.~\ref{fig:1}) given by \eqref{eq:recon_problem_geom}.
%%
%%
%The first minimization problem in \eqref{eq:recon_problem_geom} is exactly the problem given by \eqref{eq:recon_problem}. 
%The second  minimization problem in \eqref{eq:recon_problem_geom} simply determines the shortest distance between the sample-consistent plane $\Sm\fv+\Sc^\perp$ and the guiding subspace~$\Tc$. In the last minimization problem, we swap the order of  minimization.
%%
%A solution $\tv\in\Tc$ of this problem is the generalized reconstructed signal proposed in \cite{adcock2012generalized,Berger-2013}.

It can be shown that problem \eqref{eq:opt_energy} is equivalent to solving $\Am\xv = \bv$ with $\Am = \left(\Sm^\perp \Tm^\perp\right)\big| _{\Sc^\perp}$ and $\bv = - \Sm^\perp \Tm^\perp \Sm\fv$. The notation $\left(\Sm^\perp \Tm^\perp\right)\big| _{\Sc^\perp}$ makes it explicit that the domain of $\Sm^\perp \Tm^\perp$ is restricted to $\Sc^\perp$.
Conjugate gradient method (CG) is an optimal iterative method for solving  $\Am\xv = \bv$, when $\Am$ is a linear self-adjoint non-negative operator with bounded (pseudo) inverse. Although $\Sm^\perp \Tm^\perp$ is not self-adjoint in general, the restriction of $\Sm^\perp \Tm^\perp$ to $\Sc^\perp$ is self-adjoint. Therefore, 
we can use CG with $\Am = \left(\Sm^\perp \Tm^\perp\right)\big| _{\Sc^\perp}$ and $\bv = - \Sm^\perp \Tm^\perp \Sm\fv$. When the solution is not unique, CG converges to the unique normal solution $\hat{\xv}$ with minimum norm. Note that for CG to converge to the right solution, it must be initialized with some $\xv_0 \in \Sc^\perp$. 

In the special case when $\Tc+\Sc^\perp = \Hc$ and $\Tc \cap \Sc^\perp = \{\zerov\}$, 
the solution $\hat{\fv}$ of \eqref{eq:recon_problem} is the same as the result of the oblique projection $\Pm_{\Tc\perp\Sm}\fv$ in~\cite{Unser-TSP-94, Eldar-JFA-03}. But our formulation and the resulting algorithms are different since they are based only on actions of orthogonal projectors $\Tm$ and $\Sm$. 
The solution $\hat{\fv}$ of \eqref{eq:recon_problem} is closely related to generalized reconstruction given by a signal in $\Tc$ that minimizes the distance to $\Sm\fv + \Sc^\perp$.
This reconstruction is given by $\Pm_{\Tc\perp\Sm(\Tc)}\fv$, where $\Pm_{\Tc\perp\Sm(\Tc)}$ is the oblique projector onto the subspace $\Tc$ along the orthogonal complement to the subspace $\Sm(\Tc)\subseteq\Sc$~\cite{Berger-2013}. Note that generalized reconstruction does not require $\Tc+\Sc^\perp = \Hc$ (and thus, also allows for oversampling) but it may be sample inconsistent.
As illustrated in Fig.~\ref{fig:1},
it is easy to show, using \eqref{eq:recon_problem_geom},  that  $\Pm_{\Tc\perp\Sm(\Tc)}\fv=\Tm\hat{\fv}$. 
One can also show that, when the samples are noise free, reconstruction error is always better with the consistent least squares reconstruction given by \eqref{eq:recon_problem} than with the generalized reconstruction, i.e. $\|\fv-\hat{\fv}\|\leq\|\fv-\Tm\hat{\fv}\|$. We omit the details due to lack of space.

Once the consistent reconstruction $\hat{\fv}\in\Sm\fv+\Sc^\perp$ and generalized reconstruction $\tv=\Tm\hat{\fv}\in\Tc$ are known, any element in the reconstruction set, defined in Section~\ref{s:rs}, can be obtained by taking their convex combination $\alpha\hat{\fv}+(1-\alpha)\Tm\hat{\fv}$, where $\alpha \in [0,1]$. 
%as a closed interval with the end points $\hat{\fv}\in\Sm\fv+\Sc^\perp$ and $\tv=\Tm\hat{\fv}\in\Tc$.
If the samples are noise-free, we choose our reconstruction to be sample consistent,  $\hat{\fv}\in\Sm\fv+\Sc^\perp$. 
If there is noise in sample measurements, we may decide to trust the guiding subspace $\Tc$ more than the sample $\Sm\fv$ 
and choose as our output reconstruction a convex combination $\alpha\hat{\fv}+(1-\alpha)\Tm\hat{\fv}$ within the reconstruction set, 
where $0\leq\alpha<1.$ 

%\vspace{-0.5\baselineskip}
\section{Relation to Regularized Reconstruction}
\label{sec:regularization}
%\vspace{-0.25\baselineskip}
%
Regularization-based methods for reconstruction~\cite{Narang-GlobalSIP-13} can be formulated in our notation as the following unconstrained quadratic minimization problem 
\begin{equation}
\inf_{\hat{\fv}_\rho}  \; \|\Sm\hat{\fv} _\rho-\Sm\fv\|^2 +\rho\|\Hm\hat{\fv}_\rho\|^2,\quad \rho>0,
\label{eq:reg_problem_orig}
\end{equation}
where the operator $\Hm$ can be thought of as a high pass filter, e.g., it may approximate our $\Tm^\perp$,
in which case problem \eqref{eq:reg_problem_orig} approximates 
\begin{equation}
\inf_{\hat{\fv}_\rho} \; \left\|\Sm\hat{\fv}_\rho-\Sm\fv\right\|^2 +\rho\left\|\left(\hat{\fv}_\rho-\Tm\hat{\fv}_\rho\right)\right\|^2 .
\label{eq:reg_problem}
\end{equation}
Problem  \eqref{eq:reg_problem} can be viewed as a relaxation of \eqref{eq:recon_problem}. 
We prove below that the set of all solutions of \eqref{eq:reg_problem} for varying $\rho>0$ is nothing but our reconstruction set (with end points removed). 
\begin{theorem}\label{thm:reg}
Let the elements of the reconstruction set be given by the formula $\hat{\fv}_\alpha=\alpha\hat{\fv}+(1-\alpha)\Tm\hat{\fv}$, where $0\leq\alpha\leq1,$ and $\hat\fv$ solves~\eqref{eq:recon_problem}.
Then the vector $\hat{\fv}_\alpha$ is a solution of  
problem  \eqref{eq:reg_problem} with $\rho=(1-\alpha)/\alpha$.
\end{theorem}
\begin{proof}
On the one hand, minimization problem \eqref{eq:reg_problem} is equivalent to the following linear equation
$\left(\Sm+\rho\Tm^\perp\right)\hat{\fv}_\rho=\Sm\fv.$
On the other hand, the consistent reconstruction $\hat\fv$ solves
$\Sm^\perp\Tm^\perp \hat{\fv} = \zerov$ and $\Sm\hat{\fv} = \Sm\fv$. Taking $\rho=(1-\alpha)/\alpha$ and substituting
$\hat{\fv}_\alpha$ 
 for  $\hat{\fv}_\rho$, we obtain by elementary calculations
 \[
 \left(\Sm+\frac{1-\alpha}{\alpha}\Tm^\perp\right)\left(\alpha\hat{\fv}+(1-\alpha)\Tm\hat{\fv}\right) = \Sm\fv
 \]
using properties of $\Sm$ and $\Tm$ as projectors.
\end{proof}
If there exists a unique intersection of the sample-consistent plane $\Sm\fv+\Sc^\perp$ and the guiding subspace $\Tc$, as assumed in \cite{Narang-GlobalSIP-13}, then this intersection $\hat{\fv}=\Tm\hat{\fv}$. Our reconstruction set is thus trivially reduced to this single element 
$\hat{\fv}=\Tm\hat{\fv}$. By Theorem~\ref{thm:reg}, the minimizer  $\hat{\fv}_\rho$ in \eqref{eq:reg_problem} is simply
 $\hat{\fv}_\rho=\hat{\fv}=\Tm\hat{\fv}$, no matter what the value of $\rho>0$ is. 

If our reconstruction set is non-trivial, we can move the reconstruction away from the sample-consistent plane $\Sm\fv+ \Sc^\perp$ toward the guiding subspace $\Tc$. This may be beneficial when sampling is noisy.    
A specific value of the regularization parameter needs to be chosen {\it a~priori}  according to a noise level, if problem  \eqref{eq:reg_problem} is solved directly. 
Theorem \ref{thm:reg} allows us to circumvent this problem and choose $\rho$ after determining the reconstruction set. Once we have $\hat{\fv}\in\Sm\fv+\Sc^\perp$ and $\tv=\Tm\hat{\fv}\in\Tc$, we can compute the solution to \eqref{eq:reg_problem} for a desired value of $\rho$ by taking the convex combination $\alpha\hat{\fv}+(1-\alpha)\Tm\hat{\fv}$ with $\rho = (1-\alpha)/\alpha$. This also allows us to try reconstructions with multiple values of $\rho$ without solving \eqref{eq:reg_problem} each time.
Specifically, let the measurements be $\Sm\fv +\ev$, where $\ev$ denotes the noise. If the noise energy $\|\ev\|$ is known then we can select
\begin{equation}
1-\alpha=\frac{\|\ev\|}{\|\hat{\fv}-\Tm\hat{\fv}\|}. 
\label{eqn:alpha_selection}
\end{equation}

\section{Simulations}\label{s:s}
In this section, we present an application of the proposed reconstruction approach to the image magnification problem.

\vspace{-0.25\baselineskip}
\subsection{Problem setting}
\vspace{-0.2\baselineskip}
Let $\fv$ be a high resolution image of size $w \times w$. We assume that a sampled low resolution version of $\fv$ is obtained by a sampling operator $\Bm^*_\Sc$ which downsizes the image by a factor of $r$ using $r \times r$ averaging and then downsampling. Its adjoint $\Bm_\Sc$ upsamples a low resolution image by simply copying each pixel value in a $r \times r$ block to get back a $w \times w$ image. Thus, the sampling subspace $\Sc \subset \mathbb{R}^{w \times w}$ is a space of images which take a constant value in each $r \times r$ block. Note that $\dim \Sc = w/r$. The projection $\Sm\fv = \Bm_\Sc\Bm^*_\Sc \fv$ of $\fv$ on $\Sc$ is obtained replacing the values in each of its $r\times r$ blocks by their average. Our goal is to estimate $\fv$ from the input signal $\Sm\fv$. 
We know that the DCT captures most of the energy of natural images into a first few low frequency coefficients. Thus, a reasonable guiding subspace $\Tc$ is a space of images which are bandlimited to the lowest $k \times k$ frequencies. The projector $\Tm$ for this subspace is simply a low pass filter which sets the higher frequency components of the image to zero.  The projector $\Tm$ can also be decomposed as $\Bm_\Tc\Bm^*_\Tc$. Here $\Bm^*_\Tc\fv$ involves taking the DCT of $\fv$ and setting the high frequency coefficients to zero whereas $\Bm_\Tc$ converts these DCT coefficients to spatial domain to get a low frequency image. 
In our experiments, we study the effect of $\dim \Tc = k \times k$ on the quality of the reconstruction.  
We define $k_{\text{scale}} = (w / r)/k$ which compares the dimensionality of the sampling and guiding subspace. The value $k_{\text{scale}} < 1$ corresponds to an undersampling problem,  while $k_{\text{scale}} > 1$ corresponds to an oversampling scenario.
A shorthand $\fv_d$ is used to denote the low resolution image $\Bm^*_\Sc \fv$ and $\fv_{du}$ to denote the projection $\Sm\fv$.
We also consider the scenario where the samples are contaminated by noise, i.e. $\fv^n_{d} = \Bm^*_\Sc \fv + \ev$, where $\ev$ is i.i.d. Gaussian noise. As a result, the input image becomes $\fv^n_{du} = \Bm_\Sc \fv_d^n$.

\subsection{Approaches under study}

We compare four reconstruction approaches, namely, the consistent reconstruction $\hat{\fv}_c$, the generalized reconstruction $\hat{\fv}_g$, the regularized reconstruction $\hat{\fv}_\alpha$ and the minimax regret~\cite{Eldar-TSP-06} reconstruction $\hat{\fv}_m = \Tm\fv_{du}$.
The consistent reconstruction $\hat{\fv}_c$ is calculated as $\hat{\fv}_c = \hat{\xv} + \fv_{du}$, where $\hat{\xv}$ is the solution to problem $\Sm^\perp \Tm^\perp \xv = - \Sm^\perp \Tm^\perp \fv_{du}$ obtained using the CG method. 

The generalized reconstruction $\hat{\fv}_g$ is computed using three different implementations. In the first implementation, we solve the problem $\Bm^*_\Tc \Sm \Bm_\Tc \yv = \Bm^*_\Tc \fv_{du}$ using CG to obtain $\hat{\yv}$. The final reconstruction is then given by $\hat{\fv}_{g1} = \Bm_\Tc \hat{\yv}$. The second implementation uses the projector $\Tm$ instead of the sampling operator $\Bm^*_\Tc$, and the reconstruction $\hat{\fv}_{g2}$ is the CG solution to the problem $\Tm \Sm \Tm \fv = \Tm \fv_{du}$. In the third implementation, $\hat{\fv}_c$ is supposed to be available, and the generalized reconstruction is then computed by $\hat{\fv}_{g3} = \Tm \hat{\fv}_c$. Mathematically, it can be proved that all these implementations would produce identical reconstructions when the CG algorithm fully converges. However, these methods are algorithmically distinct and may converge at different rates as shown in the tests later.

The regularized reconstruction $\hat{\fv}_r$, as posed in \eqref{eq:reg_problem}, can be computed by solving $(\Sm + \rho \Tm^\perp) \fv = \fv_{du}$ via CG. If $\hat{\fv}_c$  and $\hat{\fv}_g$ are available, we can simply take the convex combination $\hat{\fv}_\alpha = \alpha \hat{\fv}_c + (1-\alpha) \hat{\fv}_g$ with $\rho = ({1-\alpha})/{\alpha}$. Because of Theorem~\ref{thm:reg}, we have $\hat{\fv}_r = \hat{\fv}_\alpha$. Although these two solutions are mathematically equivalent (upon full convergence of CG), they are not similar algorithmically and exhibit different behavior and robustness against noise, when a small fixed number of CG steps is used in our tests.

\vspace{-0.25\baselineskip}
\subsection{Experiments and observations}
\vspace{-0.2\baselineskip}
We conduct four sets of experiments to study different aspects of the reconstruction methods such as the effect of under/oversampling, effect of noise and convergence behavior. In our example, $\dim \Hc = 256 \times 256$, $\dim \Sc = 128 \times 128$.

\subsubsection{Experiment 1}
In the first experiment, we take a noise-free signal $\fv_{du}$ as input and observe the peak signal to noise ratio ($\text{PSNR} = 20 \log_{10} \frac{255}{\|\fv - \hat{\fv}\|}$) of reconstruction for different methods as the value of $k_{\text{scale}}$ (i.e. amount of under/oversampling) varies. For computing $\fv_\alpha$, we first fix $\alpha = 0.7$. Fig.~\ref{fig:e1}(a) shows the plot of PSNR against $k_{\text{scale}}$. We observe that in the undersampling regime, i.e. when $k_{\text{scale}} < 1$, $\hat{\fv}_c$ equals $\hat{\fv}_g$ and performs better than $\hat{\fv}_m$. In case of oversampling, however, Fig.~\ref{fig:e1}(a) demonstrates that $\hat{\fv}_c$ offers better PSNR than $\hat{\fv}_g$  which, in turn, performs better than $\hat{\fv}_m$.
This behavior is due to the fact that the sampling is noise-free, thus the method that keeps the samples unchanged is expected to perform better. 
The effect of $\alpha$ on the reconstruction quality is illustrated in Fig.~\ref{fig:e1}(b). Once again, we observe that as $\alpha$ increases (i.e. the samples are trusted more), the reconstruction quality improves.

\begin{figure}
%\centering
        \begin{subfigure}[b]{0.15\textwidth}\center
                \includegraphics[height=.125\textheight]{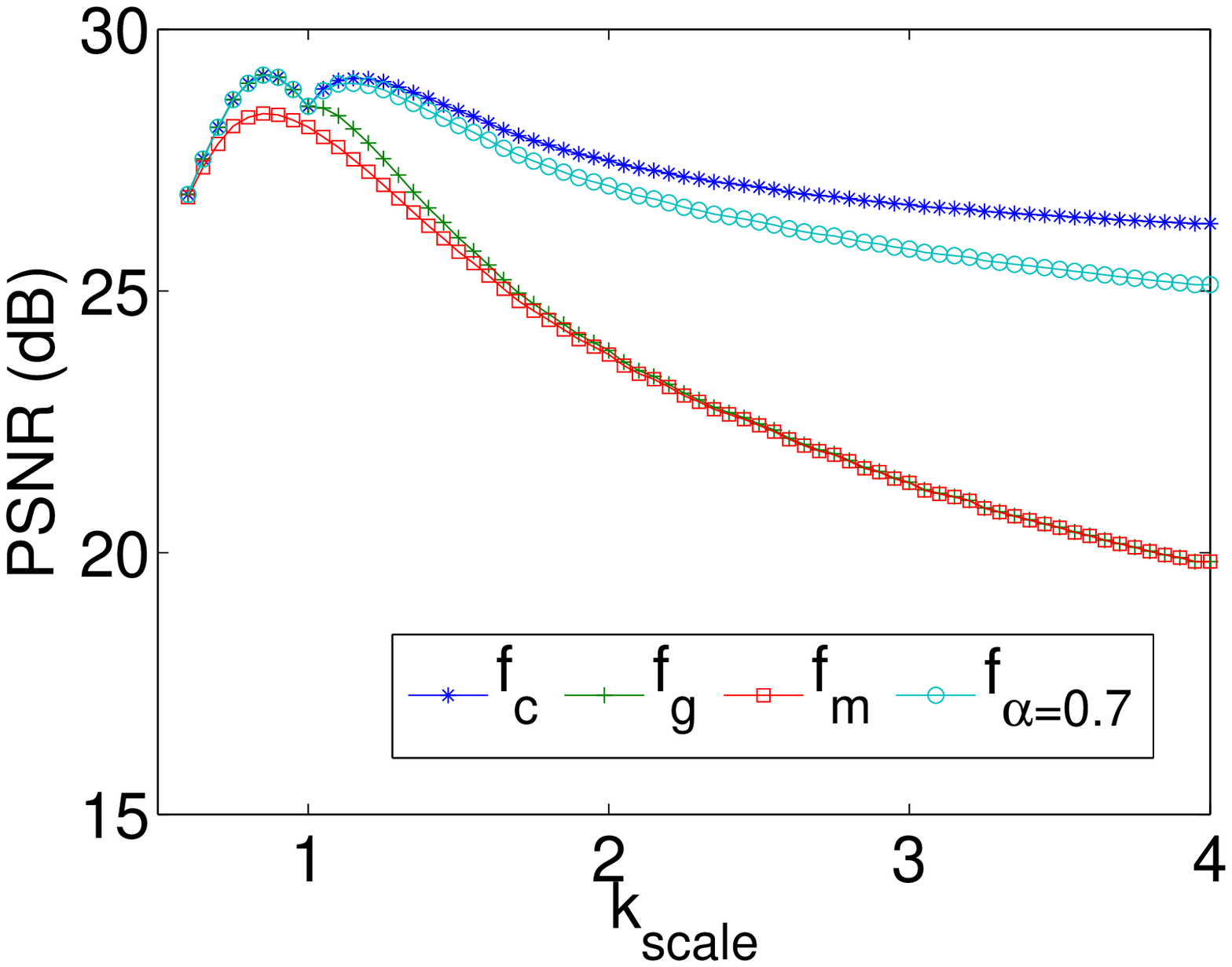}%{figs/sim/lena_k_scale.pdf}
                \caption{$\alpha=0.7$}
        \end{subfigure}
        \qquad\qquad
        \begin{subfigure}[b]{0.15\textwidth}\center
                \includegraphics[height=.125\textheight]{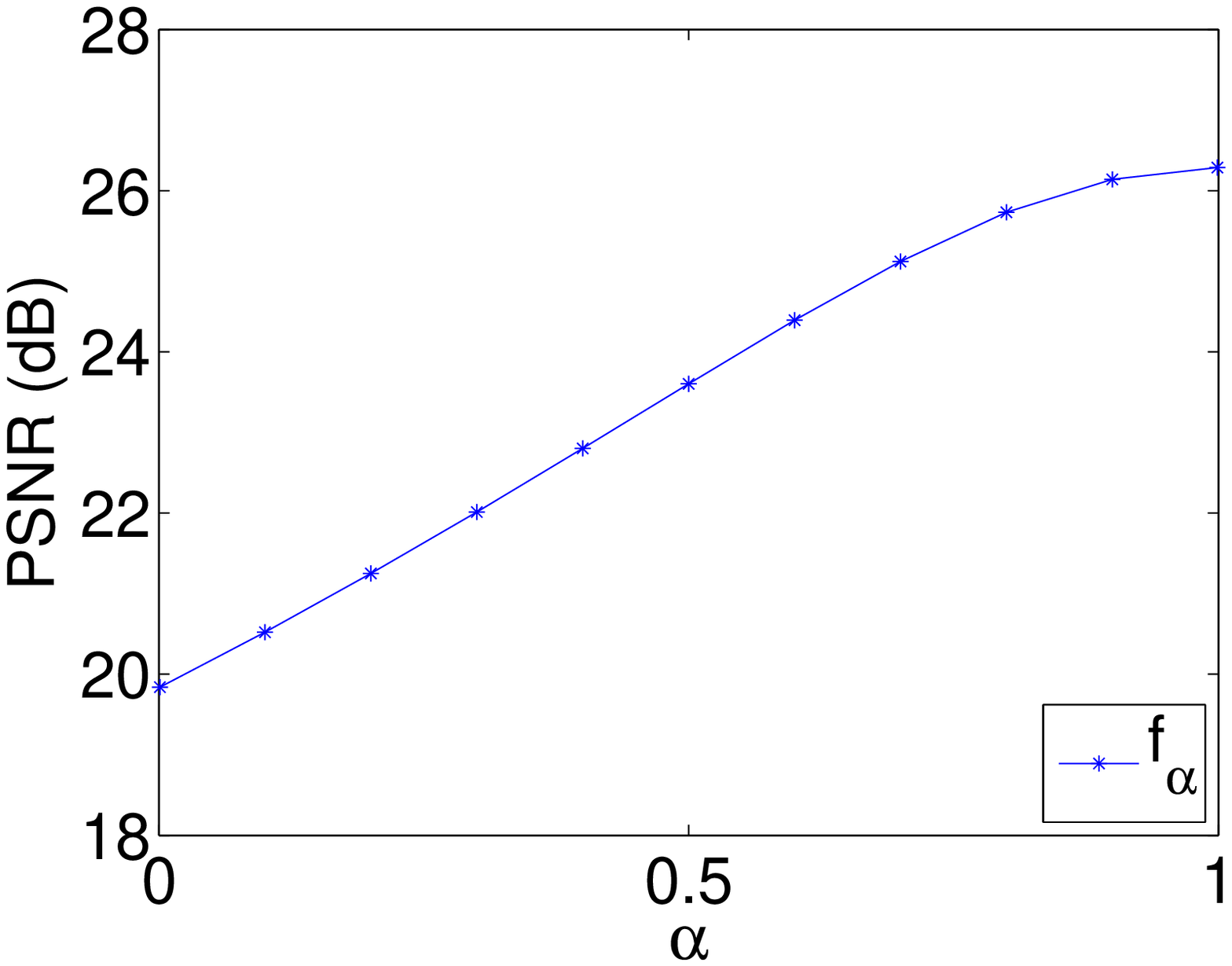}%{figs/sim/lena_alpha.pdf}
                \caption{$k_{\text{scale}} = 4$}
        \end{subfigure}
\caption{Effects of $k_{\text{scale}}$ and $\alpha$ on noise-free reconstruction}
\label{fig:e1}
\vspace{-1\baselineskip}
\end{figure}

\subsubsection{Experiment 2}
%
%For noise free inputs $\fv_{du}$, it was shown above that $\hat{\fv}_c = \hat{\fv}_{\alpha=1}$ is the best algorithm. 
In this experiment, we assume that the input $\fv_{du}^n = \Sm \fv + \ev$ is noisy, where $\ev$ is i.i.d. Gaussian  with zero mean and variance $0.001$.
We first focus on the performance of $\hat{\fv}_\alpha$ as $\alpha$ varies in case of oversampling by a factor $k_{\text{scale}} = 4$. From the results shown in Fig. \ref{fig:e3}(a), the best reconstruction is obtained with $\alpha = 0.7$. This observation agrees with the theoretically suggested optimal value $\alpha_{\text{opt}} = 1 - \|\ev\|^2/\|\hat{\fv}_g - \hat{\fv}_c\|^2 = 0.7$.  
\begin{figure}
%\centering
        \begin{subfigure}[b]{0.15\textwidth}\center
                \includegraphics[height=.125\textheight]{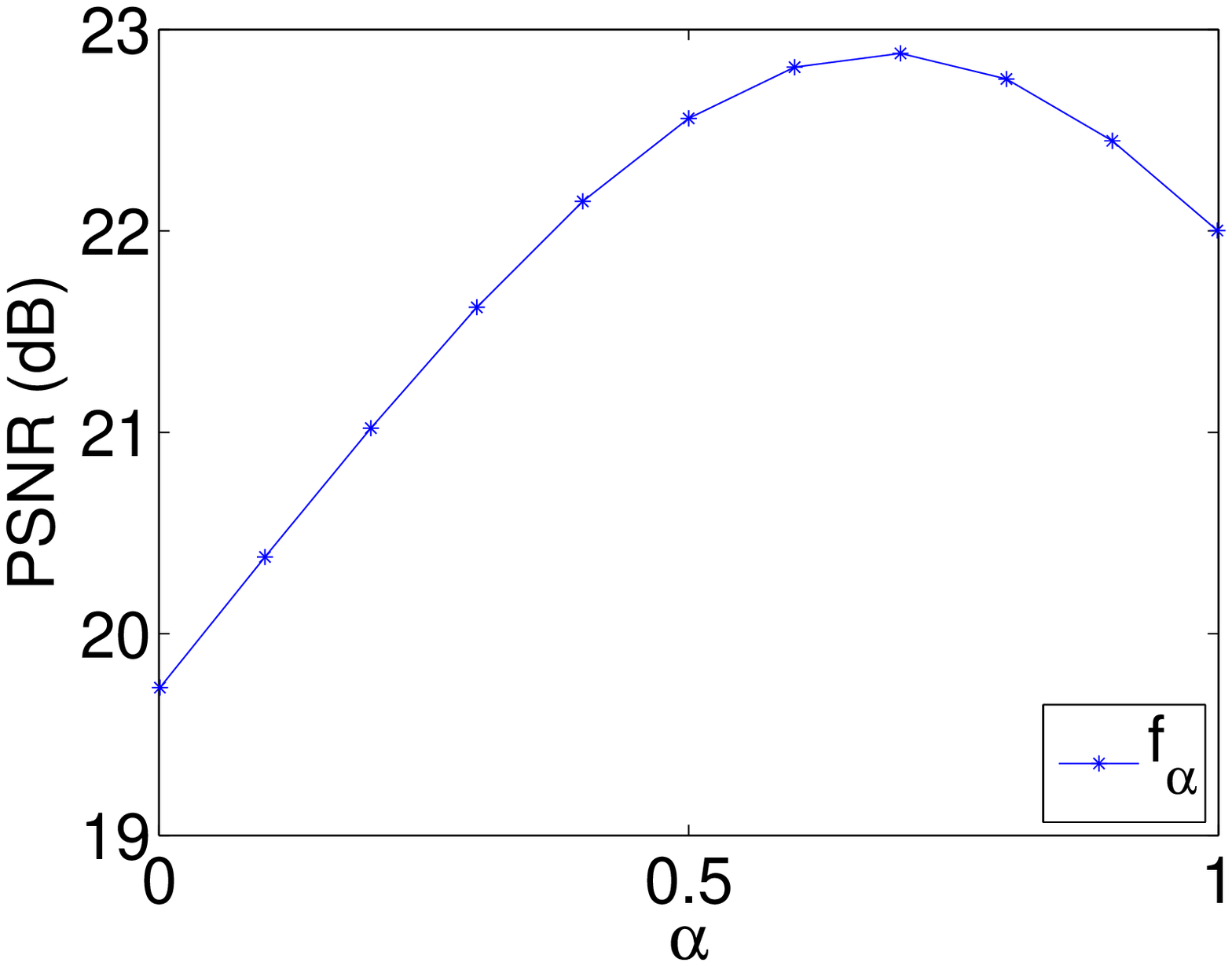}%{figs/sim/noise/lena_alpha.pdf}
                \caption{$k_{\text{scale}} = 4$}
        \end{subfigure}
        \qquad\qquad
        \begin{subfigure}[b]{0.15\textwidth}\center
                \includegraphics[height=.125\textheight]{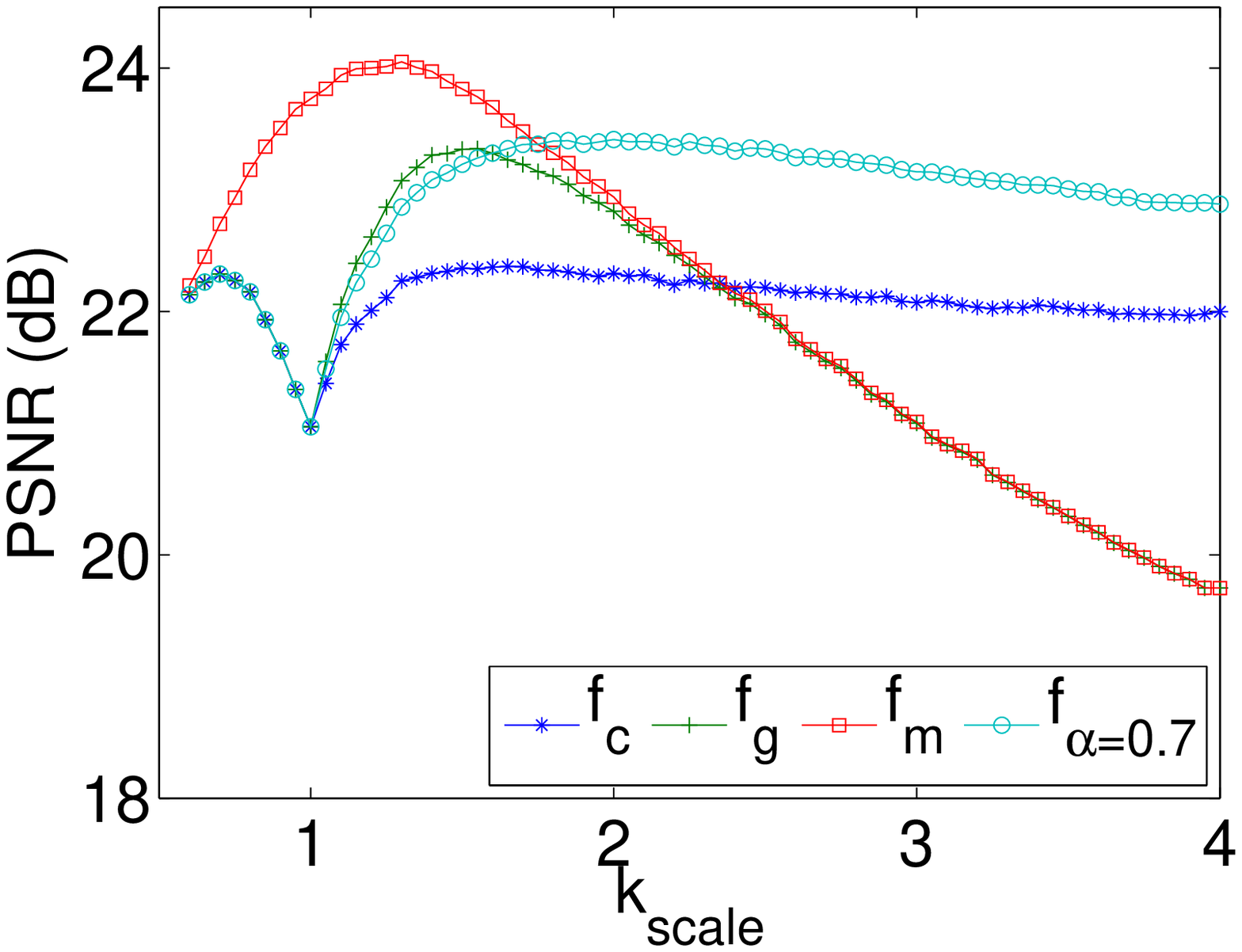}%{figs/sim/noise/lena_k_scale.pdf}
                \caption{$\alpha=0.7$}
        \end{subfigure}
\caption{Effects of $k_{\text{scale}}$ and $\alpha$ on noisy reconstruction}
\label{fig:e3}
\vspace{-1\baselineskip}
\end{figure}
Next, we analyze performance of $\hat{\fv}_g$, $\hat{\fv}_c$, $\hat{\fv}_m$, and $\hat{\fv}_{\alpha=0.7}$ for different values of $k_\text{scale}$, in Fig. \ref{fig:e3}(b). In~contrast to the previous noise-free experiment, we notice that $\hat{\fv}_c$ cannot always outperform $\hat{\fv}_g$ when noise is present. The reconstruction $\hat{\fv}_c$ only performs better than $\hat{\fv}_g$ in the heavy oversampling regime, in this example, with $k_{\text{scale}} > 2.5$. This observation indicates that, in case of slight oversampling, the noise filtering effect of the projection on guiding subspace offsets the loss due to sample inconsistency. For heavy oversampling, the sample consistency requirement is more important. 
%This implies that with noisy signal, increasing the dimension of sampling subspace will in general benefits the reconstruction quality. 
We also observe that $\hat{\fv}_\alpha$, which is a weighted combination of $\hat{\fv}_c$ and $\hat{\fv}_g$, can outperform both $\hat{\fv}_c$ and $\hat{\fv}_g$ for  $k_{\text{scale}} > ~1.5$ for this example image, since $\hat{\fv}_\alpha$ offers some noise suppression, while not deviating much from the consistency requirement. 
Fig.~\ref{fig:e4} gives an example of a noisy input image and corresponding reconstructed images, demonstrating the advantages of $\hat{\fv}_\alpha$
with $\alpha=0.7$ vs. the traditional $\hat{\fv}_g$, $\hat{\fv}_c$, and $\hat{\fv}_c$.
\begin{figure}[t]
%\vspace{-1\baselineskip}
\centering
        \begin{subfigure}[b]{0.2\textwidth}\center
                \includegraphics[height=.1\textheight]{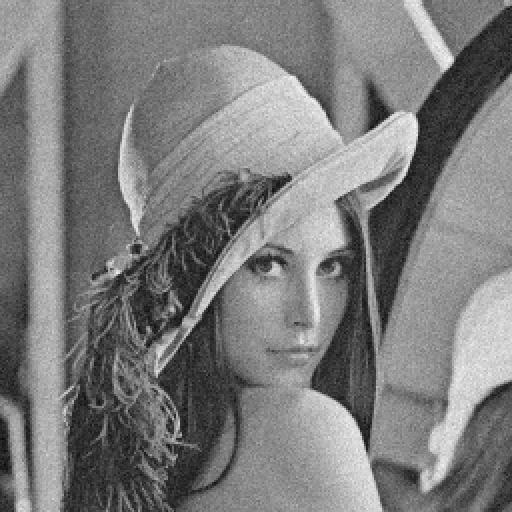}
                \caption{$\fv_{du}^*$, PSNR=21.69dB}
        \end{subfigure}
        \begin{subfigure}[b]{0.2\textwidth}\center
                \includegraphics[height=.1\textheight]{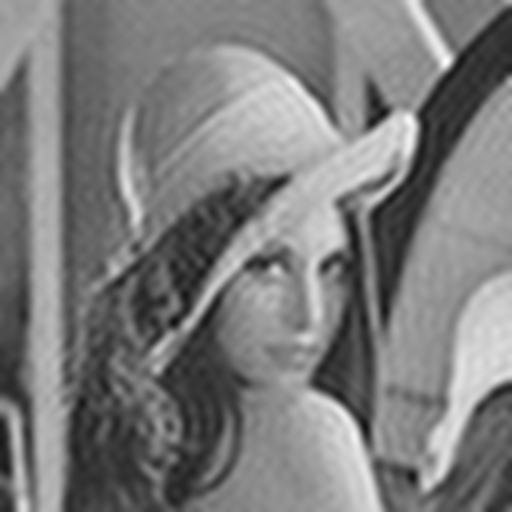}
                \caption{$\hat{\fv}_g$, PSNR=19.73dB}
        \end{subfigure}
        \begin{subfigure}[b]{0.2\textwidth}\center
                \includegraphics[height=.1\textheight]{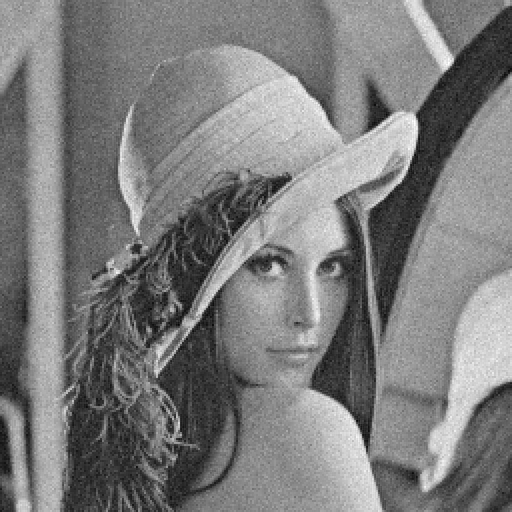}
                \caption{$\hat{\fv}_c$, PSNR=22.00dB}
        \end{subfigure}
        \begin{subfigure}[b]{0.2\textwidth}\center
                \includegraphics[height=.1\textheight]{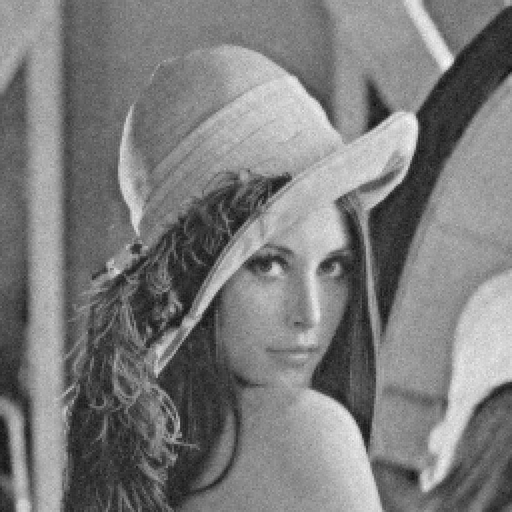}
                \caption{$\hat{\fv}_{\alpha=0.7}$, PSNR=22.88dB}
        \end{subfigure}
\caption{Reconstruction results with noisy inputs, $k_{\text{scale}} = 4$}
\label{fig:e4}
\vspace{-1\baselineskip}
\end{figure}

\subsubsection{Experiment 3}

In this experiment, we investigate the relationship between $\hat{\fv}_\alpha$ and $\hat{\fv}_r$ in case of noisy inputs. 
%since for the noise free case, there is not need to look for $\hat{\fv}_\alpha$ or $\hat{\fv}_r$, because $\hat{\fv}_c$ is always the best. 
Numerical results confirm that if the parameter $\rho$ or $\alpha$ is known beforehand and fixed, the two approaches, despite having different implementations, give identical reconstructions. 
However, if the parameter $\rho$ or $\alpha$ needs to be determined on the fly, the reconstruction $\hat{\fv}_\alpha$ is clearly favorable compared to  $\hat{\fv}_r$ in terms of computation complexity.
For determining the whole set of solutions $\{\hat{\fv}_\alpha\}$, for $\alpha \in (0, 1)$, only \emph{one} least squares problem needs to be solved, to compute $\hat{\fv}_c$. All other candidate solution points can be calculated by $\alpha \hat{\fv}_c + (1-\alpha)\Tm\hat{\fv}_c$, since $\hat{\fv}_g = \Tm \hat{\fv}_c$.
In contrast, search through the full set of $\{\hat{\fv}_r\}$, for the optimal $\rho \in (0, \infty)$, \emph{one} least squares problem needs to be solved for \emph{each} candidate solution, which may not be computationally feasible.

\subsubsection{Experiment 4}

In this experiment, we compare the performance of the reconstruction methods in terms of the number of CG iterations.
As described before, computing $\hat{\fv}_g$ can be performed in three different ways, represented by $\hat{\fv}_{g1}$, $\hat{\fv}_{g2}$, and $\hat{\fv}_{g3}$. In Fig.~\ref{fig:e5}, we compare the three implementations with number of CG iterations, $MaxIter$,  set to 1 and 2 applied to the noisy input. We observe that $\hat{\fv}_{g1} = \hat{\fv}_{g2}$ in both cases. Although $\hat{\fv}_{g3}$ is different when the number of iterations is 1, as seen in Fig.~\ref{fig:e5}(a), the difference becomes very minor when the number of iterations equals 2. This observation also holds for noise-free inputs.
\begin{figure}
%\centering
        \begin{subfigure}[b]{0.15\textwidth}\center
                \includegraphics[height=.125\textheight]{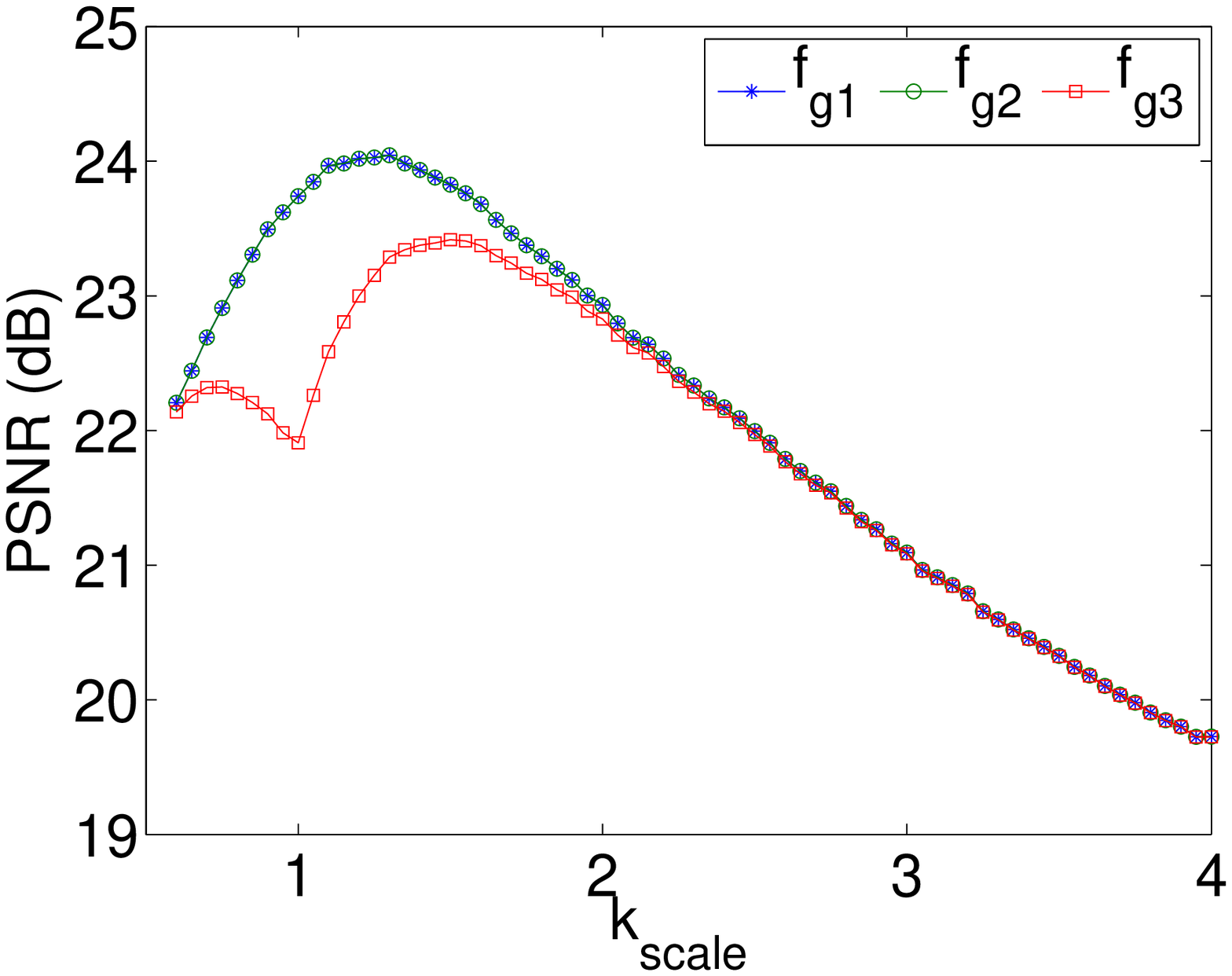}%{figs/sim/schedule4/gs_iter1.pdf}
                \caption{$MaxIter=1$}
        \end{subfigure}
        \qquad\qquad
        \begin{subfigure}[b]{0.15\textwidth}\center
                \includegraphics[height=.125\textheight]{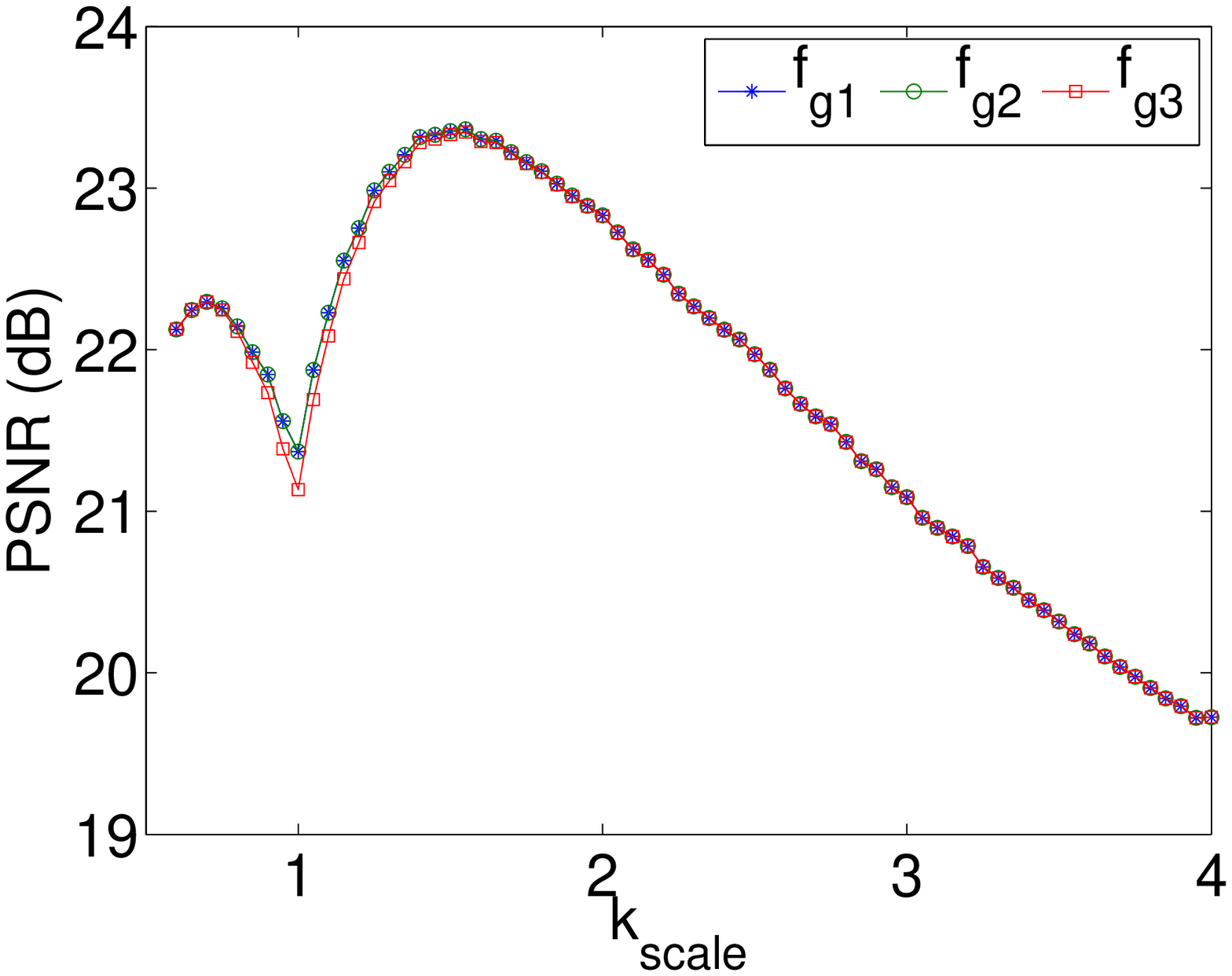}%{figs/sim/schedule4/gs_iter2.pdf}
                \caption{$MaxIter=2$}
        \end{subfigure}
\caption{Performance of different implementations of $\hat{\fv}_{g}$.}
\label{fig:e5}
\vspace{-1\baselineskip}
\end{figure}
Since $\hat{\fv}_g$ has three implementations, $\hat{\fv}_\alpha$ can also have different corresponding implementations, given by $\hat{\fv}_{\alpha i}=\alpha \hat{\fv}_c + (1-\alpha) \hat{\fv}_{gi}$ with $i=1,2,3$. The performance of all the reconstruction methods with different implementations is shown in Fig. \ref{fig:e6}. All the algorithms are configured to use $MaxIter$ number of CG iterations, except $\hat{\fv}_m$ since it does not need least squares. The reconstruction $\hat{\fv}_{\alpha 2}$ is omitted as it is always equal to $\hat{\fv}_{\alpha 1}$. 
We observe that $\hat{\fv}_{\alpha1} = \hat{\fv}_{\alpha2}$ performs better than $\hat{\fv}_{\alpha3}$. In case of heavier oversampling, $\hat{\fv}_\alpha$ is more favorable than~$\hat{\fv}_r$. 
\begin{figure}[h]
\vspace{0\baselineskip}
        \begin{subfigure}[b]{0.15\textwidth}\center
                \includegraphics[height=.125\textheight]{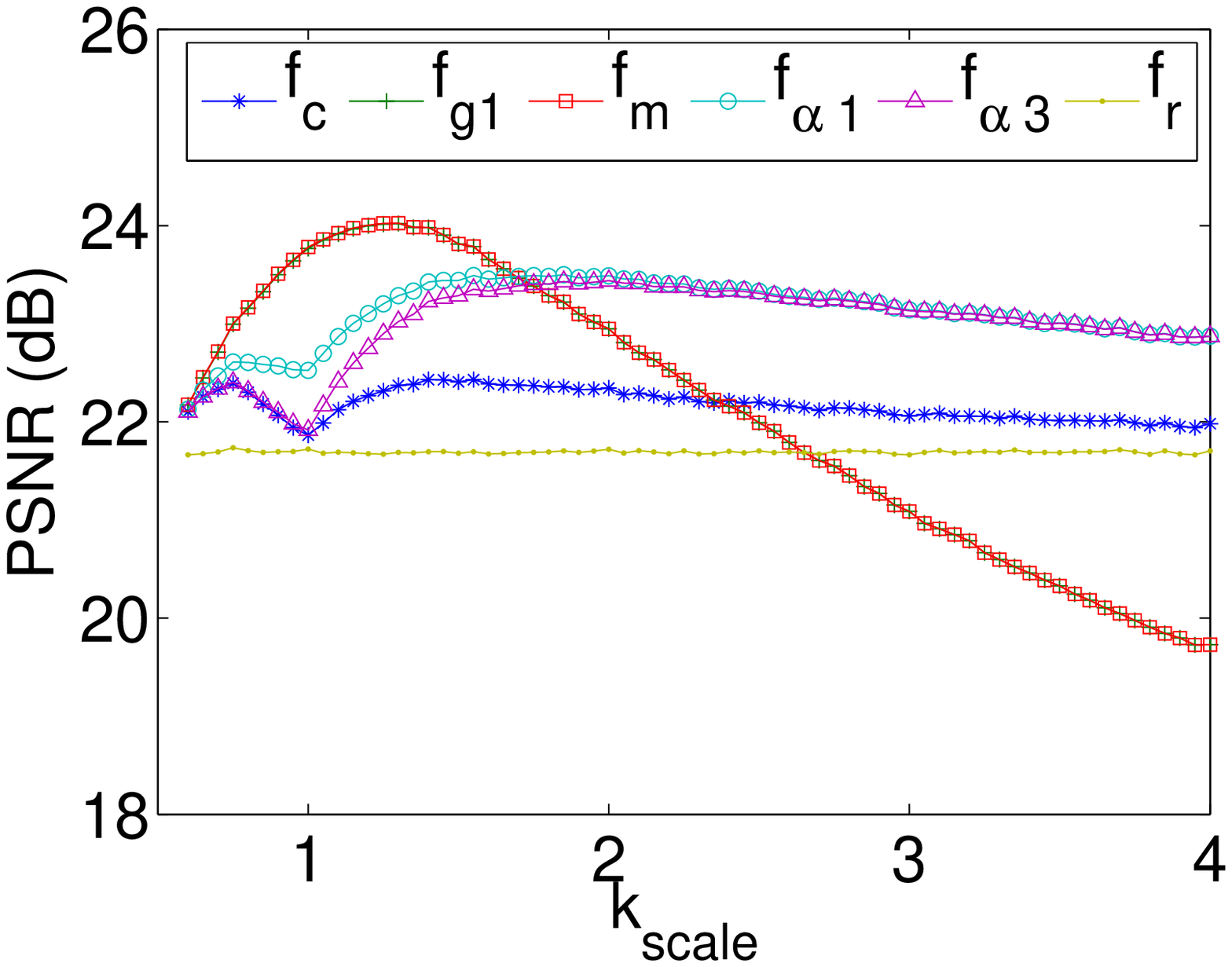}%{figs/sim/schedule4/two_alpha_iter1.pdf}
                \caption{$MaxIter=1$}
        \end{subfigure}
        \qquad\qquad
        \begin{subfigure}[b]{0.15\textwidth}\center
                \includegraphics[height=.125\textheight]{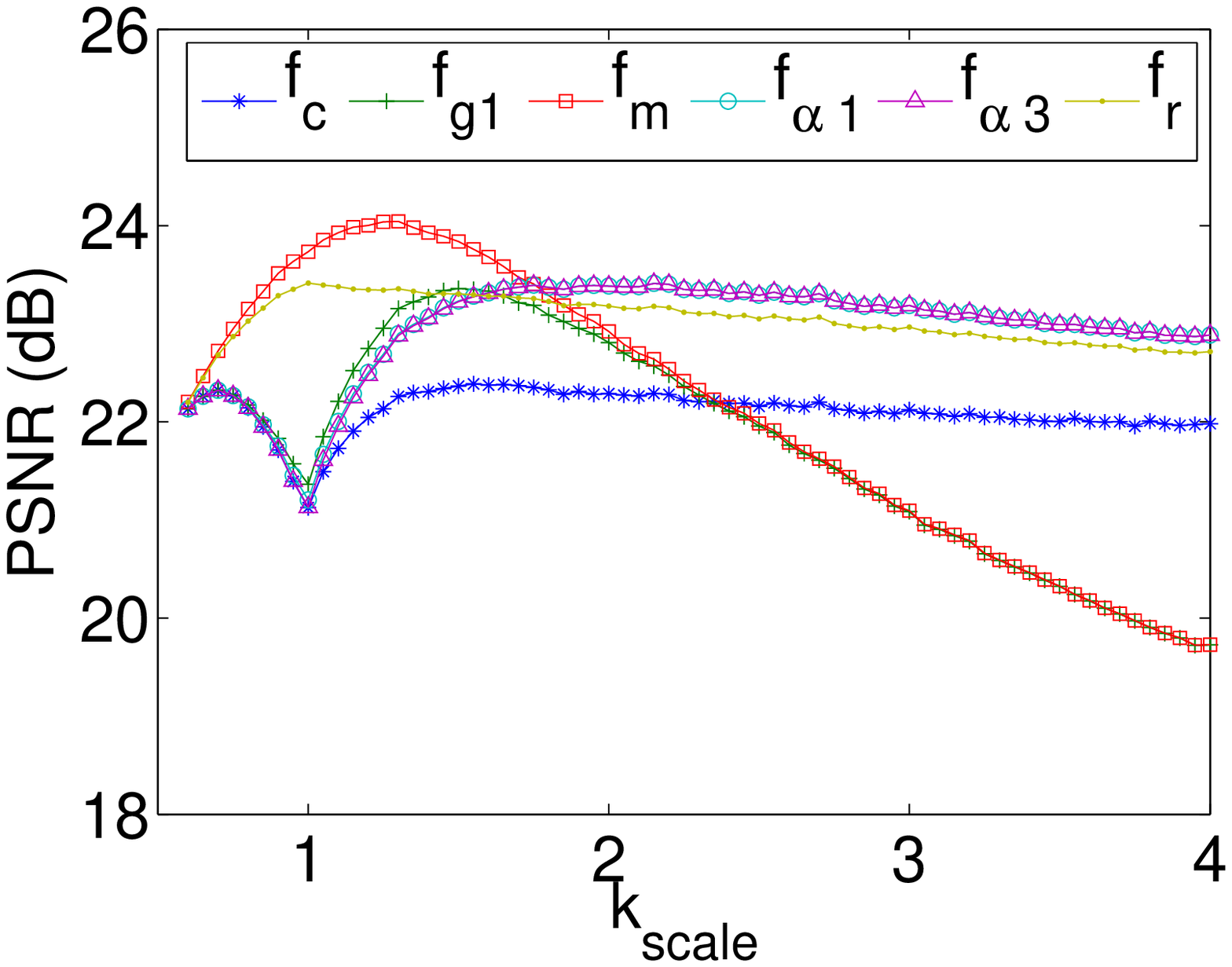}%{figs/sim/schedule4/two_alpha_iter2.pdf}
                \caption{$MaxIter=2$}
        \end{subfigure}
\caption{Reconstructed image qualities}
\label{fig:e6}
\vspace{-1\baselineskip}
\end{figure}

%\vspace{-0.25\baselineskip}
\section*{Conclusion}
%\vspace{-0.25\baselineskip}
The proposed frame-less iterative algorithm allows efficiently reconstructing noisy signals with desired properties described by a guiding subspace, 
which can be given by an approximate projector. 
Numerical examples for image magnification demonstrate the advantages of our method, compared to the traditional  
sample consistent and pure guided methods.
The suggested methodology is expected to be effective for a wide range of signal reconstruction applications, in video and speech processing, and machine learning.  

\newpage
\IEEEtriggeratref{7}
\bibliographystyle{IEEEtran}
\bibliography{refs}

%$^*$A preliminary version of this work has been posted on arxiv.org

\end{document}